\documentclass[11pt,a4paper]{article}
\usepackage{amsmath,amsthm,amssymb,amsfonts,mathrsfs,eucal,amsthm}
\usepackage{graphicx,hyperref}
\usepackage[stable]{footmisc}
\usepackage{graphics}
\usepackage{epsfig}
\usepackage{amsmath}
\usepackage{amssymb}
\newtheorem{definition}{Definition}
\newtheorem{theorem}{Theorem}

\newcommand\sH{{\cal{H}}}
\newcommand\eps{{\varepsilon}}
\newcommand\states{\mathfrak{D}}
\newtheorem{lemma}{Lemma}

\newcommand{\hrho}{\hat{\rho}}
\newcommand{\hvarrho}{\hat{\varrho}}
\newcommand{\hsigma}{\hat{\sigma}}
\newcommand{\homega}{\hat{\omega}}
\newcommand{\hI}{\hat{I}}
\newcommand{\be}{\begin{equation}}
\newcommand{\ee}{\end{equation}}
\newcommand{\bea}{\begin{eqnarray}}
\newcommand{\eea}{\end{eqnarray}}
\newcommand\tr{\operatorname{Tr}}
\begin{document}
\title{Entanglement Cost for Sequences of Arbitrary Quantum States}
\author{Garry Bowen \footnote{Centre for Quantum Computation, DAMTP, University of Cambridge, Cambridge CB3 0WA, UK} and 
Nilanjana Datta\footnote{Statistical Laboratory, University of Cambridge,    Wilberforce Road, Cambridge CB3 0WB, UK    
(e-mail:n.datta@statslab.cam.ac.uk)}}
\maketitle

\begin{abstract}
The entanglement cost of arbitrary sequences of bipartite states is shown to be expressible as the minimization of a conditional spectral entropy rate over sequences of separable extensions of the states in the sequence.  The expression is shown to reduce to the regularized entanglement of formation when the $
n^{th}$ state in the sequence consists of $n$ copies of a single bipartite state.
\end{abstract}

\section{Introduction} 

A fundamental problem in entanglement theory is to determine how to optimally convert
entanglement, shared between two distant parties Alice and Bob, from one form to 
another. Entanglement manipulation is the process by which Alice and Bob convert 
an initial bipartite state $\rho_{AB}$ which they share, to a required
target state $\sigma_{AB}$ using local operations and classical communication (LOCC). 
If the target state $\sigma_{AB}$ is a maximally entangled state, then the protocol 
is called entanglement distillation, whereas if the initial state $\rho_{AB}$ is a
maximally entangled state, then the protocol is called entanglement dilution. Optimal
rates of these protocols were originally evaluated under the assumption that the 
entanglement resource accessible to Alice and Bob consist of multiple copies, i.e.,
tensor products $\rho_{AB}^{\otimes n}$, of the initial bipartite state $\rho_{AB}$, and the requirement
that the final state of the protocol is equal to $n$ copies of the desired target state 
$\sigma_{AB}^{\otimes n}$ with asymptotically vanishing error in the limit $n\rightarrow \infty$.
The distillable entanglement and entanglement cost computed
in this manner are two asymptotic measures of entanglement of the
state $\rho_{AB}$. Moreover, in the case in which $\rho_{AB}$ is pure,
these two measures of entanglement coincide and are equal to the von
Neumann entropy of the reduced state on any one of the subsystems, $A$
or $B$.

In this paper we focus on entanglement dilution, which, as mentioned earlier, is the entanglement manipulation process by which 
two distant parties, say Alice and Bob, create a desired bipartite 
target state from a maximally entangled state which they initially share,
using LOCC. In \cite{bennett96b,bennett96} the optimal rate
of entanglement dilution, namely, the {\em{entanglement cost}}, was evaluated
in the case in which Alice and Bob created multiple copies of a
desired target state $\rho_{AB}$, with asymptotically vanishing error,
from a shared resource of singlets, using local operations and classical communication. In particular, for the case of a pure target state $\rho_{AB}$, 
the entanglement cost was shown \cite{bennett96} to be
equal to the {\em{entropy of entanglement}} of the state, i.e., the von Neumann entropy of the reduced state on any one of the two subsystems $A$ and $B$. 
Moreover, in \cite{hayden00} it was shown that 
for an arbitrary mixed state $\rho_{AB}$, the entanglement cost is
equal to the regularized entanglement of formation of the state (see 
(\ref{eof_reg}) for its definition).

The practical ability to transform entanglement from one form to
another is useful for many applications in quantum information
theory. However, it is not always justified to assume that the
entanglement resource available consists of states which are multiple
copies (and hence tensor products) of a given entangled state, or to require that the final state of the
protocol is of the tensor product form. More
generally an entanglement resource is characterized by an arbitrary
sequence of bipartite states which are not necessarily of the tensor
product form.  Sequences of bipartite states on $AB$ are considered to
exist on Hilbert spaces $\sH_A^{\otimes n} \otimes\sH_B^{\otimes n}$
for $n \in \{1,2,3 \ldots\}$. 
%%%%%%% end insertion %%%%%%%%%%%%

A useful tool for the study of entanglement manipulation in this general scenario is provided by the \textit{Information Spectrum} method.  
The information spectrum method, introduced in classical information theory by Verdu \& Han \cite{verdu94,han}, has been extended into quantum information theory by Hayashi, Ogawa and Nagaoka \cite{ogawa00,nagaoka02,hayashi03}.  The power of the information spectrum approach comes from the fact that it does not depend on 
the specific structure of sources, channels or entanglement resources employed in information theoretical protocols.

In this paper we evaluate the optimal asymptotic 
rate of {\em{entanglement dilution}} for an arbitrary sequence of bipartite states. The
case of an arbitrary sequence of pure bipartite state was studied in \cite{pure}. The paper is organized as follows. 
In Section~\ref{notations} 
we introduce the necessary notations and definitions. 
Section~\ref{main} contains the statement and proof of the main result, stated as Theorem~\ref{cost_theorem}. Note that if we consider the sequence of bipartite states to consist of tensor products of a given bipartite state, then our main
result reduces to the known results obtained in \cite{bennett96b, hayden00}
(see the discussion after Theorem \ref{cost_theorem} of Section \ref{main}).
Finally, in Section~\ref{asymp} we
show how Theorem~\ref{cost_theorem} yields an alternative
proof of the equivalence of the asymptotic entanglement cost and the regularised entanglement of formation~\cite{hayden00}.

\section{Notations and definitions}
\label{notations}

Let ${\mathcal B}({\mathcal H})$ denote the algebra of linear operators acting on a finite--dimensional Hilbert space ${\mathcal H}$ of dimension $d$ and let $\states(\sH)$ denote the set of states (or density operators, i.~e. positive operators of unit trace) acting on $\sH$.
Further, let $\sH^{(n)}$ denote the Hilbert space $\sH^{\otimes n}$. For any state $\rho \in \states(\sH)$, the von Neumann entropy is defined as $S(\rho):= - \tr \left(\rho \log \rho\right)$.

Let $\Lambda^n$ be a quantum operation used for the transformation of an initial bipartite state $\omega^n $ to a bipartite state $\rho^n$, with
$\omega^n, \rho^n \in \states\left((\sH_A\otimes \sH_B)^{\otimes n}\right)$.  For the entanglement manipulation processes considered in this paper, $\Lambda^n$ either consists of local operations (LO) alone or LO with one-way or two-way classical communication.
We define the efficacy of any entanglement manipulation process in terms of the fidelity $F_n := \mathrm{Tr} \sqrt{\sqrt{\rho^n}\Lambda^n(\omega^n)
\sqrt{\rho^n}}$ between the output state $\Lambda^n(\omega^n)$ and the target state $\rho^n$.
An entanglement manipulation process is said to be \textit{reliable} if the asymptotic fidelity $\mathcal{F} := \liminf_{n\rightarrow \infty} F_n = 1$.

For given orthonormal bases $\{|\chi^i_A\rangle\}_{i=1}^{d^n_A}$ and
$\{|\chi^i_B\rangle\}_{i=1}^{d^n_B}$ in Hilbert spaces ${\mathcal{H}}_A^{\otimes n}$ and
 ${\mathcal{H}}_B^{\otimes n}$, of dimensions $d^n_A$ and $d^n_B$ respectively, we define
the canonical maximally entangled state of Schmidt rank $M_n \le \min\{d^n_A, d^n_B\}$
to be 
\be
|\Psi^{M_n}_{AB}\rangle = \frac{1}{\sqrt{M_n}}\sum_{i=1}^{M_n} 
|\chi^i_A\rangle \otimes |\chi^i_B\rangle.
\ee 
In fact, in the following, we consider ${\mathcal{H}}_A\simeq {\mathcal{H}}_B$,
for simplicity, so that ${d^n_A}={d^n_B}$. Here and henceforth, the explicit
$n$-dependence of the basis states $|\chi^i_A\rangle$ and $|\chi^i_B\rangle$
has been suppressed for notational simplicity.
%\subsection{Spectral Projections}

The quantum information spectrum approach requires the extensive use of spectral projections. Any self-adjoint operator $A$ acting on a finite dimensional Hilbert space may be written in its spectral
decomposition $A = \sum_i \lambda_i \pi_i$, where $\pi_i$ denotes the operator
which projects onto the eigenspace corresponding to the eigenvalue $\lambda_i$.  We define the
positive spectral projection on $A$ as $\{ A \geq 0 \} = \sum_{\lambda_i \geq 0} \pi_i$, i.e., the projector onto the eigenspace of positive eigenvalues of $A$.  For two operators $A$ and $B$, we can
then define $\{ A \geq B \}$ as $\{ A - B \geq 0 \}$.  The following key lemmas are used repeatedly in the paper. For their proofs
see \cite{ogawa00,nagaoka02}.
\begin{lemma}
\label{lemma}
For self-adjoint operators $A$, $B$ and any positive operator $0 \leq P \leq I$
the inequality
\begin{equation}
\mathrm{Tr}\big[ P(A-B)\big] \leq \mathrm{Tr}\big[ \big\{ A \geq B \big\}
(A-B)\big]
\label{eqn:first_ineq}
\end{equation}
holds.
\end{lemma}
\begin{lemma}
\label{cor0}
Given a state $\rho^n \in  \states(\sH^{\otimes n})$ and a self-adjoint
operator $\omega^n \in {\cal{B}}(\sH^{\otimes n})$, we have
\begin{equation}
\mathrm{Tr}\big[\{\rho^n \ge e^{n\gamma}\omega^n \} \omega^n \bigr]
\leq e^{-n\gamma}.
\end{equation}
for any real number $\gamma$.
\end{lemma}

In the quantum information spectrum approach one defines spectral divergence
rates, which can be viewed as generalizations of the quantum relative entropy.
The spectral generalizations of the von Neumann entropy, the conditional
entropy and the mutual information can all be expressed as spectral divergence rates.
\begin{definition}
Given a sequence of states $\hrho=\{\rho^n\}_{n=1}^\infty$, with
$\rho^n \in  \states(\sH^{\otimes n})$, and
a sequence of positive operators $\homega=\{\omega^n\}_{n=1}^\infty$,
with $\omega^n \in {\cal{B}}(\sH^{\otimes n})$, 
the quantum spectral sup-(inf-) divergence rates are defined in terms
of the difference operators $\Pi^n(\gamma) = \rho^n - e^{n\gamma}\omega^n$, 
for any arbitrary real number $\gamma$, as
\begin{align}
\overline{D}(\hrho \| \homega) &= \inf \Big\{ \gamma : \limsup_{n\rightarrow \infty} \mathrm{Tr}\big[ \{ \Pi^n(\gamma) \geq 0 \} \Pi^n(\gamma) \big] = 0 \Big\} \nonumber \\
\underline{D}(\hrho \| \homega) &= \sup \Big\{ \gamma : \liminf_{n\rightarrow \infty} \mathrm{Tr}\big[ \{ \Pi^n(\gamma) \geq 0 \} 
\Pi^n(\gamma) \big] = 1 \Big\} \nonumber
\end{align}
respectively.
\end{definition}
The spectral entropy rates and the conditional spectral entropy rates
can be expressed as divergence rates with appropriate
substitutions for the sequence of operators $\homega = \{ \omega^n
\}_{n=1}^{\infty}$.  These are $\overline{S}(\hrho) = -\underline{D}(\hrho\| \hI)$ and $\underline{S}(\hrho) = -\overline{D}(\hrho\| \hI)$, where 
$\hI=\{I^n\}_{n=1}^\infty$, with $I^n$ being the identity operator acting on the Hilbert space ${\mathcal{H}}^{\otimes n}$. Further, for sequences of bipartite states $\hrho_{AB} = \{\rho_{AB}^n\}_{n=1}^\infty$,
\begin{align}
\overline{S}(A|B) &= -\underline{D}(\hrho_{AB}\| \hI_{A}\otimes \hrho_B) \\
\underline{S}(A|B) &= -\overline{D}(\hrho_{AB}\| \hI_{A}\otimes \hrho_B) \; .
\end{align}
In the above,
$\hI_{A}=\{I_A^n\}_{n=1}^\infty$ and $\hrho^{B}=\{\rho_B^n\}_{n=1}^\infty$,
with $I_A^n$ being the identity operator in ${\mathcal{B}}({\mathcal{H}}_A^{(n)})$
and $\rho_B^n = \mathrm{Tr}_A \rho_{AB}^n$, the partial trace
being taken on the Hilbert space ${\mathcal{H}}_A^{(n)}$. Various properties of these quantities, and relationships between them, are explored in \cite{bowen06}.

For sequences of states $\hrho=\{\rho^{\otimes n}\}$ and $\homega=\{\omega^{\otimes n}\}$, with $\rho, \omega \in \states(\sH)$,
it has been proved \cite{hayashi03} that
\be\label{stein}
\overline{D}(\hrho \| \homega)=\underline{D}(\hrho \| \homega)
={S}(\rho \| \omega),
\ee
where $S(\rho\|\omega) := \tr \rho \log \rho - \tr \rho \log \omega$, is the
quantum relative entropy.

Two parties, Alice and Bob, share a sequence of maximally entangled states
$\{|\Psi^{M_n}_{AB}\rangle\}_{n=1}^\infty$, and wish to convert them into a sequence of given bipartite states $\{\rho^n_{AB}\}_{n=1}^\infty$, with
$\rho^n \in  \states(\sH^{\otimes n})$ and 
$|\Psi^{M_n}_{AB}\rangle \in {\mathcal{H}}^{\otimes n}_A \otimes 
{\mathcal{H}}_B^{{\otimes n}}$.  The protocol used for this conversion is known as
entanglement dilution. The concept of reliable entanglement manipulation may then be used to define an asymptotic entanglement measure, namely the {\em{entanglement cost}}.
\begin{definition}
A real-valued number ${\mathcal{R}}$ is said to be an \textit{achievable} dilution rate for a sequence of states $\hrho_{AB}=\{\rho^n_{AB}\}$, with $\rho^n_{AB} \in  \states((\sH_A \otimes \sH_B)^{\otimes n})$, if $\forall \eps > 0,\, \exists N$ such that $\forall n \geq N$ a transformation exists that takes $|\Psi_{AB}^{M_n}\rangle \langle \Psi_{AB}^{M_n}| \rightarrow \rho^n_{AB}$ with fidelity $F_n^2 \geq 1-\eps$ and $\frac{1}{n}\log M_n \leq {\mathcal{R}}$.
\end{definition}
\begin{definition}
The \textit{entanglement cost} of the sequence $\hrho_{AB}$ is the infimum of all achievable dilution rates:
\begin{equation}
E_C(\hrho_{AB}) = \inf {\mathcal{R}}
\end{equation}
\end{definition}

To simplify the expressions representing the entanglement cost, we define the following sets of sequences of states.  Firstly, given a sequence of target states $\hrho_{AB} = \{ \rho^n_{AB} \}_{n=1}^{\infty}$, define the set $\mathcal{D}_{cq}(\hrho_{AB})$ as the set of sequences of tripartite states $\hvarrho_{RAB} = \{ \varrho^n_{RAB} \}_{n=1}^{\infty}$ 
such that each
$\varrho_{RAB}^n$ is a \textit{classical-quantum state} (\textit{cq}-state) of the form
\begin{equation}\label{var}
 \varrho_{RAB}^n = \sum_i p_i^{(n)} |i^n_R \rangle \langle i^n_R| \otimes |\phi_{AB}^{n,i}\rangle \langle \phi_{AB}^{n,i}|,
\end{equation}
where $\rho_{AB}^n = \sum_i p_i^{(n)} |\phi_{AB}^{n,i}\rangle \langle \phi_{AB}^{n,i}|$ and the set of pure states $\{ |i^n_R\rangle \}$ form an orthonormal 
basis of ${\mathcal{H}}_R^{\otimes n}$. We refer to the state $ \varrho_{RAB}^n$ as a cq-extension of the bipartite state $ \rho^n_{AB}$. Let  
$\mathcal{D}_{cq}^n(\rho_{AB}^n)$ denote the set of all possible cq-extensions of $ \rho^n_{AB}$.   
\medskip

The {\em{entanglement of formation}} of the bipartite state $\rho_{AB}^n \in 
{\states}(({\mathcal{H}}_A \otimes{\mathcal{H}}_B)^{\otimes n})$ is defined as
$$
E_F(\rho_{AB}^n):= \min_{\{p_i^{(n)}, |\phi_{AB}^{n,i}\rangle \}}
\sum_i p_i^{(n)} S(\rho^{n,i}_A),
$$
where $\rho^{n,i}_A = \tr_B |\phi_{AB}^{n,i}\rangle\langle\phi_{AB}^{n,i}|$,
the partial trace being taken over the Hilbert space ${\mathcal{H}}_B^{\otimes n}$,
and the minimization is over all possible ensemble decompositions of the state
$\rho^n_{AB}$.
Alternatively, the entanglement of formation of the state $\rho_{AB}^n$ can be expressed as
\be
\label{eof}
E_F(\rho_{AB}^n)= \min_{\mathcal{D}_{cq}^n(\rho_{A B}^n)} S(A|R)_{\varrho_{RA}^n},
\ee
where $S(A|R)_{\varrho_{RA}^n}$ denotes the conditional entropy
$$S(A|R)_{\varrho_{RA}^n} = S(\varrho_{RA}^n) - S(\varrho_R^n),$$
with $\varrho_{RA}^n = \tr_B \varrho_{RAB}^n$ and $\varrho_{R}^n = \tr_A \varrho_{RA}^n$, the state $\varrho_{RAB}^n$ being a cq-extension of the state $\rho_{AB}^n$.  

The {\em{regularized entanglement of formation}} of a bipartite state $\rho_{AB} \in 
{\states}({\mathcal{H}}_A \otimes{\mathcal{H}}_B)$ is defined as
\be
E_F^\infty(\rho_{AB}) := \lim_{n\rightarrow \infty} \frac{1}{n} E_F(\rho_{AB}^{\otimes n})
\label{eof_reg}
\ee

The sup-conditional entropy rate $\overline{S}(A|R)$ of the sequence 
$\hvarrho_{RA}:=\{\varrho_{RA}^n\}_{n=1}^\infty$, defined as
\be\label{supar}
\overline{S}(A|R) = -\underline{D}(\hvarrho_{RA}\| \hI_{A}\otimes 
\hat{\varrho_R}),
\ee
where $\hat{\varrho_R}:= \{\varrho_{R}^n\}_{n=1}^\infty$, will be of particular significance in this paper.
\medskip

\noindent
{\em{Note:}} For notational simplicity, the explicit $n$-dependence of quantities are suppressed in the rest of the paper, wherever there is no scope of any ambiguity.

%\subsection{Quantum Spectral Information Rates}
%\label{infspectrum}

\section{Entanglement Dilution for Mixed States}
\label{main}

The asymptotic optimization over entanglement dilution protocols leads to the following theorem.
\begin{theorem}
\label{cost_theorem}
The entanglement cost of a sequence of bipartite target states $\hrho_{AB} = \{ \rho_{AB}^n \}_{n=1}^{\infty}$, is given by
\begin{equation}
E_C(\hrho_{AB}) = \min_{{\mathcal{D}}_{cq}(\hrho_{AB})}\overline{S}(A|R),
\end{equation}
or equivalently $\min_{{\mathcal{D}}_{cq}(\hrho_{AB})}\overline{S}(B|R)$, where ${{\mathcal{D}}_{cq}(\hrho_{AB})}$ is the set of sequences of tripartite states
$\hvarrho_{RAB} = \{ \varrho^n_{RAB} \}_{n=1}^{\infty}$ defined above.
\end{theorem}

The proof of Theorem \ref{cost_theorem} is contained in the following two lemmas. However, before going over to the proof, we would first like to point out 
that previously known results on entanglement dilution \cite{bennett96, hayden00} 
can be recovered from the above theorem. In \cite{hayden00} 
it was
proved that the entanglement cost of an arbitrary (mixed) bipartite state 
$\rho_{AB}$, evaluated in the case in which Alice and Bob create {\em{multiple copies}} (i.e., tensor products) of $\rho_{AB}$ (with asymptotically vanishing error, from a shared resource of singlets, using LOCC) is given by the regularized entanglement of formation $E_F^\infty(\rho_{AB})$ (\ref{eof_reg}). In Section 
\ref{asymp} we prove how this result can be recovered from Theorem \ref{cost_theorem}. As regards the entanglement cost of pure states, in \cite{pure} we obtained an expression for the entanglement cost of an arbitrary sequence of 
pure states and we proved that this expression reduced to the entropy of entanglement of a given pure state (say, $|\psi_{AB}\rangle$), if the sequence consisted
of tensor products of this state -- thus recovering the result first proved in
\cite{bennett96b}.

\begin{lemma} {\em{(Coding)}}
For any sequence $\hrho_{AB}=\{\rho_{AB}^n\}_{n=1}^\infty$ and $\delta > 0$, the dilution rate
\begin{equation}
{\mathcal{R}}= \overline{S}(A|R)+\delta,
\end{equation}
where $\overline{S}(A|R)$ is the sup-conditional spectral rate given by (\ref{supar}),
is achievable.
\end{lemma}
\begin{proof}
Let the target bipartite state $\rho_{AB}^n$ have a decomposition given by
\begin{equation}
\rho_{AB}^n = \sum_i p_i |\phi^i_{AB} \rangle \langle \phi^i_{AB}|,
\end{equation}
where the Schmidt decomposition of $|\phi^i_{AB} \rangle $ is given by 
\begin{equation}
\label{schmidt}
|\phi^i_{AB} \rangle = \sum_{k} \sqrt{\lambda^{i,k}} |\psi^{i,k}_{A}\rangle
|\psi^{i,k}_{B}\rangle,
\end{equation}
with the Schmidt coefficients $\lambda_k^i$ being arranged in non-increasing order, i.e., $\lambda_1^i \ge \lambda_2^i \ldots \ge \lambda_{d_n}^i$, for $d_n = {\rm{dim }}\, {\mathcal{H}}_{A}^{\otimes n}$.

Alice locally prepares the classical-quantum state (cq-state) $\,\rho_{RAA^\prime}^n = \sum_i p_i |i_R\rangle\langle i_R| \otimes |\phi^i_{AA^\prime} \rangle \langle \phi^i_{AA^\prime}| \in {\states}\left(({\mathcal{H}}_R\otimes 
{\mathcal{H}}_A\otimes {\mathcal{H}}_{A^\prime})^{\otimes n}\right)$. She then 
does a unitary operation on the system $RAA^\prime$ given by
$$\bigl(I_A^{n} \otimes \Theta^n_{RA^\prime}\bigr) \,\rho_{RAA^\prime}^n \,\bigl(I_A^{n} \otimes \Theta^n_{RA^\prime}\bigr)^{-1},
$$
where
\begin{equation}
\Theta^n_{RA^\prime} := \sum_j |j_R\rangle \langle j_R| \otimes \sum_l 
|\chi^l_{A^\prime}\rangle \langle \psi^{j,l}_{{A^\prime}}|,
\end{equation}
with $\{|\chi^l_{A^\prime}\rangle\}_{l=1}^{d_n}$ being a fixed orthonormal basis in ${\mathcal{H}}_{A^\prime}^{\otimes n} $. This results in 
the state
\begin{equation}
 \sum_i p_i |i_R\rangle\langle i_R| \otimes  \sum_{k} \sqrt{\lambda^{i,k}\,
\lambda^{i,k'}}
 |\psi^{i,k}_{A}\rangle\langle \psi^{i,k'}_{A}|\otimes |\chi^k_{{A^\prime}}\rangle\langle \chi^{k'}_{{A^\prime}}|,
\end{equation}
where, once again, the explicit $n$-dependence of the terms has been suppressed
for notational simplicity.
Note that Alice's operation amounts to a coherent implimentation of a projective measurement on $R$ with rank one
projections $|j_R\rangle\langle j_R|$, followed by a unitary $U_j = \sum_l 
|\chi^l_{A^\prime}\rangle \langle \psi^{j,l}_{{A^\prime}}|$
on $A^\prime$, conditional on the outcome $j$.
Alice teleports the $A^\prime$ state to Bob.  The resultant
shared state is
\begin{equation}
\nu^n_{RAB}:= \sum_i p_i |i_R\rangle\langle i_R| \otimes  \sum_{k, k'=1}^{M_n} \sqrt{
\lambda^{i,k}\,\lambda^{i,k'}}
|\psi^{i,k}_{A}\rangle\langle \psi^{i,k'}_{A}|\otimes |\chi^k_{{B}}\rangle\langle \chi^{k'}_{{B}}| + \sigma_{RAB}^n
\end{equation}
where $\sigma_{RAB}^n$ is an unnormalized error state.
Note that the sum over the index $k$ is truncated to $M_n$. This truncation occurs due to the so-called {\em{quantum scissors effect}} \cite{qsc}, i.e., 
if the quantum state to be teleported lives in a space of dimension 
higher than the rank $M_n$ of the shared entangled state (used by the two
parties for teleportation), then all higher-dimensional terms in the expansion
of the original state are cut off. Moreover the system $A^\prime$ is now referred to as $B$, since it is now in Bob's possession.

 Alice also sends the ``classical'' state $R$ to Bob through a 
classical channel. Bob then acts on the system $RB$, which is now in his 
possession, with the unitary operator $(\Theta_{RB}^n)^\dagger$.
The final shared state can therefore be expressed as
\bea
\bigl(I_A^n \otimes (\Theta_{RB}^n)^\dagger\bigr)\nu^n_{RAB}\bigl(I_A^n \otimes \Theta_{RB}^n\bigr)
&=& \omega_{RAB}^n + \tilde{\sigma}_{RAB}^n \nonumber\\
&:=& \sum_i p_i |i_R\rangle\langle i_R| \otimes | \hat{\phi}^{i}_{AB}\rangle \langle \hat{\phi}^{i}_{AB} | + \tilde{\sigma}_{RAB}^n,\nonumber\\
\eea
where
\begin{equation}
| \hat{\phi}^{i}_{AB}\rangle := (Q_A^{M_n,i}\otimes I^{B})|\phi^{i}_{AB}\rangle \; ,
\end{equation}
with $Q_A^{M_n,i}$ being the orthogonal projector onto span of the Schmidt vectors
corresponding to the 
${M_n}$ largest Schmidt coefficients of $|\phi^i_{AB}\rangle$, and $\tilde{\sigma}_{RAB}^n:= (\Theta_{RB}^n)^\dagger\sigma_{RAB}^n  \Theta_{RB}^n.$

By Uhlmann's theorem (see \cite{nielsen}) it follows that $F (\omega_{AB}^n + \tilde{\sigma}_{AB}^n, \rho_{AB}^n) \geq F (\omega_{AB}^n, \rho_{AB}^n)$ \footnote{Take a purification $|\omega\rangle$ such that $F(\omega,\rho) = \langle \omega|\rho\rangle$.  Then utilize purifications $|\omega\rangle |0\rangle$, $|\sigma\rangle |1\rangle$, and $|\rho\rangle |0\rangle$, which along with Uhlmann's theorem implies $F(\omega+\sigma,\rho) \geq \langle \omega|\rho\rangle = F(\omega,\rho)$.} and the fidelity between the state $\omega_{AB}^n$ of the entanglement dilution protocol and the target state
$\rho_{AB}^n$ is bounded below by
\begin{equation}
F_n \geq \max_{|\rho_{ABC}^n\rangle} \big| \langle \rho_{ABC}^n|\omega_{ABC}^n\rangle \big|,
\label{fid0}
\end{equation}
where $|\omega_{ABC}^n\rangle$ is any fixed purification of the final state $\omega_{AB}^n$ and the maximization is taken over all purifications of $\rho_{AB}^n$.

By choosing purifications $|\omega_{CAB}^n\rangle = \sum_i \sqrt{p_i} |i_C\rangle |\hat{\phi}^{i}_{AB}\rangle$ and 
$|\rho_{CAB}^n\rangle = \sum_i \sqrt{p_i} |i_C\rangle |\phi^{i}_{AB}\rangle$, we obtain the 
following lower bound to $F_n^2(\omega^n_{AB}, \rho^n_{AB})$. Let 
$Q^n_{RA} := \sum_i |i_R\rangle \langle i_R| \otimes Q_A^{M_n,i}$ and $\rho^n_{RA} := 
\sum_i p_i |i_R\rangle \langle i_R| \otimes \rho_A^{n,i}$, where $\rho_A^{n,i} = \mathrm{Tr}_B |\phi^{i}_{AB} \rangle 
\langle \phi^{i}_{AB}|$. Then 
\begin{align}
F_n^2 \geq  \big|\langle \omega_{CAB}^n|\rho_{CAB}^n\rangle \big|^2 &= \mathrm{Tr} \big[ Q^n_{RA} \rho^n_{RA} \big]\nonumber\\
&= \sum_i p_i \mathrm{Tr} \big[ Q^{M_n,i}_{A} \rho^{n,i}_A\big]
\label{one1}
\end{align}
Explicitly examining the projection operator $P_{RA}^n := \{\sum_i p_i |i_R\rangle \langle i_R| \otimes \rho^{n,i}_A \geq e^{-n\alpha} \rho_R^n \otimes I_A^n \}$, where $\alpha$ is a real number, we can express it in the form $P_{RA}^n = \{\sum_i p_i |i_R\rangle \langle i_R| \otimes \big( \rho^{n,i}_A - e^{-n\alpha} I_A^n \big) \geq 0 \} = \sum_i |i_R\rangle \langle i_R| \otimes \{ \rho_A^{n,i} \geq e^{-n\alpha}I_A \}$.  The rank of each of the projectors $\{ \rho_A^{n,i} \geq e^{-n\alpha}I_A \}$ is then bounded by $\mathrm{Tr} [\{ \rho_A^{n,i} \geq e^{-n\alpha}I_A \}] \leq e^{n\alpha}$ by Lemma \ref{cor0}, and hence by comparing $P^n_{RA}$ with $Q^n_{RA}$ we can see that $M_n = \lceil e^{n\alpha} \rceil$ implies that $\mathrm{Tr} [ Q_{RA}^n \rho^n_{RA} ] \geq \mathrm{Tr} [ P_{RA}^n \rho^n_{RA} ]$.
For any $\delta > 0$ we can always choose a positive integer $N$ such that for all $n \geq N$ there is an integer $M_n$ satisfying $\overline{S}(A|R) + \delta \geq \frac{1}{n}\log M_n > \overline{S}(A|R)$. Thus, using a sequence of maximally entangled states $\{|\Psi^{M_n}_{AB}\rangle\}_{n=1}^\infty$ of Schmidt rank $M_n$, from the definition of $\overline{S}(A|R)$ it follows that
\begin{equation}
\mathcal{F}^2 \geq \lim_{n\rightarrow \infty} \mathrm{Tr}\big[ Q_{RA}^n \rho^n_{RA} \big] \geq \lim_{n\rightarrow \infty} \mathrm{Tr}\big[ P_{RA}^n \rho^n_{RA} \big] = 1.
\end{equation}
and entanglement dilution at the rate ${\mathcal{R}} = \overline{S}(A|R) + \delta$ is achievable.
\end{proof}

\begin{lemma} {\em{(Weak Converse)}}
For any arbitrary sequence of states $\hrho_{AB}$, any entanglement dilution protocol with a rate
\begin{equation}
 {\mathcal{R}}^* < \min_{\mathcal{D}} \overline{S}(A|R),
\end{equation}
where $\overline{S}(A|R)$ is the sup-conditional spectral rate given by (\ref{supar}), 
is not reliable.
\end{lemma}
\begin{proof}
Let ${\mathcal{T}}_{AB}^n$ denote any LOCC operation used for transforming
the maximally entangled state $|\Psi^{M_n}_{AB}\rangle \in {\mathcal{H}}_A^{\otimes n}  
\otimes
{\mathcal{H}}_B^{\otimes n} $ to the target state $\rho_{AB}^n$ in this Hilbert space,
such that $F\Bigl({\mathcal{T}}_{AB}^n\big(|\Psi^{M_n}_{AB}\rangle \langle\Psi^{M_n}_{AB}| \big) , \rho^n_{AB} \Bigr) \rightarrow 1$ as $n \rightarrow \infty$.
Employing the Lo \& Popescu theorem \cite{lo99}, the final state of the
protocol is expressible as
\bea
\omega_{AB}^n &:=& {\mathcal{T}}_{AB}^n\big(|\Psi^{M_n}_{AB}\rangle \langle\Psi^{M_n}_{AB}| \big) \nonumber\\
&=& \sum_k(K_A^{n,k} \otimes U_B^{n,k})|\Psi^{M_n}_{AB}\rangle \langle\Psi^{M_n}_{AB}| (K_A^{n,k} \otimes U_B^{n,k})^\dagger
\eea
with $\sum_k (K_A^{n,k})^{\dag} K_A^{n,k} = I_A^n$, and $U_B^{n,k}$ is unitary.

Let $|\omega_{CAB}^n\rangle := \sum_k |k^n_C\rangle \otimes (K_A^{n,k} \otimes U_B^{n,k})|\Psi^{M_n}_{AB}\rangle$, 
denote a purification of the final state, $\omega_{AB}^n$, of the entanglement dilution
protocol, with $C$ denoting a reference system, and $\{ |k^n_C\rangle\}$ denoting an orthonormal basis
in its Hilbert Space ${\mathcal{H}}_C^{\otimes n} $. By Uhlmann's theorem, for this
fixed purification $|\omega_{CAB}^n\rangle$, the fidelity is given by 
\begin{equation}
F_n(\rho^n_{AB}, \omega^n_{AB}) = \max_{|\rho_{CAB}^n\rangle} |\langle \rho_{CAB}^n| \omega_{CAB}^n\rangle|,
\label{fid}
\end{equation} 
where the maximization is over all purifications $|\rho_{CAB}^n\rangle$ of the
target state $\rho_{AB}^n$. 
However, this maximization is equivalent 
to a maximization over all possible unitary transformations acting on the 
reference system $C$. This in turn corresponds to a particular
decomposition of the purification of the target state $\rho_{AB}^n$ with 
respect to a fixed reference system \cite{hughston93}.
Explicitly we then have $|\rho^n_{CAB}\rangle = \sum_k \sqrt{p_k^{(n)}} |k^n_C\rangle |\phi^{n,k}_{AB}\rangle$,
where 
$\sum_k p_k^{(n)} |\phi_{AB}^{n,k}\rangle \langle \phi_{AB}^{n,k} |$ is the given decomposition of $\rho_{AB}^n$
obtained from the maximization. 

Then 
\be\label{23}
F_n(\rho^n_{AB}, \omega^n_{AB}) = |\langle \rho_{CAB}^n| \omega_{CAB}^n\rangle| = 
|\sum_k \sqrt{p_k} \langle \phi_{AB}^{n,k} | K^{n,k}_A \otimes U^{n,k}_B |\Psi^{M_n}_{AB}\rangle|,
\ee
Note that
\begin{align}
 U_B^{n,k}|\Psi^{M_n}_{AB}\rangle &= \frac{1}{\sqrt{{M_n}}} \sum_{j=1}^{M_n} |\chi_A^j\rangle
 U_B^{n,k}|\chi^j_B\rangle
\nonumber\\
&= \frac{1}{\sqrt{{M_n}}}  P_A^{M_n}\sum_{j=1}^{N_n} |\chi_A^j\rangle
 U_B^{n,k}|\chi^j_B\rangle,
\label{xx}
\end{align} 
where  $N_n = \dim \mathcal{H}^{\otimes n}_A$ and $P_A^{M_n} = \sum_{k = 1}^{M_n}|\chi^k_A\rangle \langle \chi^k_A|$.

For simplicity, let us consider $\sH_A \simeq \sH_B \simeq \sH$, and let the state $|\phi^{n,k}_{AB}\rangle \in (\sH_A \otimes \sH_B)^{\otimes n} \simeq \sH^{\otimes 2n}$ have a Schmidt decomposition 
$|\phi^{n,k}_{AB}\rangle = \sum_i \sqrt{\lambda^{n,k}_{ i}} 
|\psi^{n,k}_{A,i}\rangle\otimes |\psi^{n,k}_{B,i}\rangle$. Further, let 
$W^k_A$ and $W^k_B$ be unitary operators in ${\cal{B}}(\sH^{\otimes n})$
such that
$W^k|\psi_{A,j}^{n,k}\rangle = |\chi^j_{A}\rangle$ and $W_k|\psi_{B,j}^{n,k}\rangle = |\chi^j_{B}\rangle$. Then from (\ref{xx}) it follows that
\bea
 U_B^{n,k}|\Psi^{M_n}_{AB}\rangle &=&  \frac{1}{\sqrt{{M_n}}}  P_A^{M_n}\sum_{j=1}^{N_n} W^k_A|\psi_{A,j}^{n,k}\rangle
 U_B^{n,k} W^k_B|\psi_{B,j}^{n,k}\rangle,\nonumber\\
&=& \frac{1}{\sqrt{{M_n}}}P_A^{M_n} \sum_{j=1}^{N_n} V_A^{n,k} |\psi_{A,j}^{n,k}\rangle|\psi_{B,j}^{n,k}\rangle 
\eea
where $V_A^{n,k} := (U^{n,k}_B W^k_B)^T W^k_A$.
Here we have used the relation $\sum_j |j\rangle \otimes U|j\rangle = \sum_j U^T |j\rangle \otimes |j\rangle$ for $U$ unitary and $\{ |j\rangle \}$ an orthonormal basis in $\mathcal{H}^{\otimes n}_A$.

  Then
\begin{equation}
\langle \phi^{n,k}_{AB}| K^{n,k}_A \otimes U^{n,k}_B |\Psi^{M_n}_{AB}\rangle
= \mathrm{Tr}\Big[ \frac{1}{\sqrt{{M_n}}} {\sqrt{\rho_{A}^{n,k}}}K_A^{n,k} P_A^{M_n} V_A^{n,k}  \Big] \nonumber
\end{equation}
where $\rho^{n,k}_{A} = \mathrm{Tr}_B  |\phi^{n,k}_{AB}\rangle\langle \phi^{n,k}_{AB} | = 
\sum_i \lambda^{n,k}_{i} |\psi^{n,k}_{A,i}\rangle \langle \psi^{n,k}_{A,i} |$.
Then from (\ref{23}), using the Cauchy Schwarz inequality we then obtain
\begin{align}
F_n &= \Big| \mathrm{Tr}\bigl[ \sum_k \frac{1}{\sqrt{{M_n}}}{\sqrt{p_{n,k}
\rho^{n,k}_{A}}}K_A^{n,k} P_A^{M_n} V_A^{n,k}
\bigr]\Big| \nonumber\\
&\leq \sum_k \Big| \mathrm{Tr}\bigl[ \frac{1}{\sqrt{{M_n}}}K_A^{n,k} P_A^{M_n}\cdot P_A^{M_n} V_A^{n,k} {\sqrt{p_{n,k}\rho^{n,k}_{A}}}
\bigr]\Big| \nonumber\\
&\leq \sum_k  \Big(\frac{1}{{{M_n}}}\mathrm{Tr}\bigl[P_A^{M_n} (K_A^{n,k})^\dagger K_A^{n,k}\bigr] \cdot
\mathrm{Tr}\bigl[p_{n,k} \rho_{A}^{n,k} P_A^{n,k}\bigr] \Big)^{\frac{1}{2}} \nonumber \\
&\leq \Big( \sum_k p_{n,k} \mathrm{Tr} \big[ P_A^{n,k} \rho_{A}^{n,k} \big]\Big)^{\frac{1}{2}}
\label{last}
\end{align}
where $P_A^{n,k} = (V^{n,k}_A)^\dagger P_A^{M_n} V_A^{n,k}$.  The third inequality follows by the following argument.  Express the third line as $\sum_k \sqrt{q_k \pi_k}$ with $q_k = \frac{1}{{{M_n}}}\mathrm{Tr}\bigl[P_A^{M_n} (K_A^{n,k})^\dagger K_A^{n,k}\bigr]$ and $\pi_k = \mathrm{Tr} \big[ P_A^{n,k} p_k \rho_{A}^{n,k} \big] \geq 0$.  From the properties $\mathrm{Tr}[P_A^{M_n}] ={M_n}$ and $\sum_i (K^{n,k}_A)^{\dag} K^{n,k}_A = I_A^n$ it follows that $\sum_k q_k = 1$ and $q_k \geq 0$ for all $k$. Then using the concavity of the map $x \mapsto \sqrt{x}$, we have that
\bea  \sqrt{\sum_k \pi_k}\ge    \sqrt{\sum_{k\,:\,{q_k>0}} q_k \Bigl(\frac{\pi_k}{q_k}}\Bigr)&\ge & \sum_{k\,:\,{q_k>0}} q_k \sqrt{\frac{\pi_k}{q_k}}\nonumber\\
&=& \sum_k \sqrt{{q_k}}\sqrt{{\pi_k}},
\eea
yielding the inequality in the last line of (\ref{last}).

Defining the projection operator $${P}_{RA}^{n}:= \sum_j |j_R^n\rangle \langle j_R^n| \otimes P_A^{n,j},$$ and the state
$$\rho_{RA}^n:= \sum_k p_{n,k} |k_R^n\rangle \langle k_R^n| \otimes \rho^{n,k}_A,$$
the square of the fidelity can then be bounded by
\bea
F^2_n &\leq &  \mathrm{Tr}\big[{P}_{RA}^{n}\rho_{RA}^n \big]\label{conv}\\
&\le & \sum_n p_{n,k} \mathrm{Tr} \big[ Q^{M_n,k}_{A} \rho^{n,k}_A\big],
\label{two2}
\eea
where $Q_A^{M_n,k}$ is the orthogonal projector onto the span of the Schmidt vectors
corresponding to the ${M_n}$ largest Schmidt coefficients of $|\phi^{n,k}_{AB}\rangle$.

Note that eqs.~(\ref{one1}) and (\ref{two2}) yields an alternative proof of the following lemma stated in \cite{haya_book}:
\begin{lemma}
\label{lemma_fidelity}
The entanglement dilution fidelity for a given bipartite state $\rho^n_{AB}:= \sum_i p_i |\phi^i_{AB}\rangle
\langle \phi^i_{AB}|$, under an LOCC transformation $\Lambda^n$ is given by
\begin{align}
F^2 (\Lambda^n(\Psi^{M_n}_{AB}), \rho^n_{AB})  &= \sum_i p_i \mathrm{Tr} \big[ Q^{M_n,i}_{A} \rho^{n,i}_A\big]\nonumber\\
&= \sum_i p_i \sum_{j=1}^{M_n} \lambda_j^i,
\end{align}
where $\lambda_j^i$, $j=1, \ldots, M_n$ denote the $M_n$ largest Schmidt coefficients of $|\phi^i_{AB}\rangle$.
\end{lemma}

From (\ref{conv}), using Lemma \ref{lemma}, with $\Pi^n(\gamma) := \rho_{RA}^n - e^{-n\gamma}\rho_R^n \otimes I_A^n$,
\begin{align}
\mathrm{Tr}\big[P_{RA}^{n} \rho_{RA}^n\big] &= \mathrm{Tr}\big[ P_{RA}^{n} \Pi^n(\gamma) \big] + e^{-n\gamma}\sum_k p_k \mathrm{Tr}[P_A^{n,k}] \nonumber \\
&\leq \mathrm{Tr}\big[ \{  \Pi^n(\gamma) \geq 0 \}\Pi^n(\gamma) \big] + {M_n}e^{-n\gamma} \nonumber
\end{align}
since $\mathrm{Tr}[P_{n,k}]=\mathrm{Tr}[P_{{M_n}}]={M_n}$. Hence for ${M_n} \le e^{n\mathcal{R}}$ we have
\begin{equation}
F^2_n \leq {\mathrm{Tr}}\big [\{ \Pi^n(\gamma) \geq 0 \}\Pi^n(\gamma) \big] + e^{-n(\gamma - \mathcal{R})}.
\label{last2}
\end{equation}
Choosing a number $\gamma$ and $\delta > 0$ such that ${\mathcal{R}} + \delta = \gamma <
\overline{S}(A|R)$, the second term on RHS of (\ref{last2}) tends to zero as $n
\rightarrow \infty$. However, since $\gamma < \overline{S}(A|R)$
the first term on RHS of (\ref{last2}) does not converge to $1$ as
$n \rightarrow \infty$. Hence, the asymptotic fidelity ${\mathcal{F}}$
is not equal to $1$.

It is then straightfoward to show that the particular choice of decomposition of each $\rho_{AB}^n$ imposed by the fidelity criterion gives a minimization over possible \textit{cq}-sequences.  Suppose there exists a {cq}-sequence $\hsigma_{RAB}$ with $\overline{S}_{\sigma}(A|R) = \overline{S}_{\rho}(A|R) - \eps$ for some $\eps > 0$.  It then follows from the coding theorem that the rate ${\mathcal{R}} = \overline{S}_{\sigma}(A|R) + \eps/2$ is asymptotically attainable.  However, if we take $F'_n = |\langle \sigma_{RR'AB}^n | \omega_{RR'AB}^n \rangle|$ then this is less than the maximization over all possible purifications, bounding the asymptotic fidelity below 1, giving a contradiction.

%The minimization over sequences in $\mathcal{D}_{cq}$ can then be extended to a minimization over all separable sequences 
%$\mathcal{D}$ simply by noting that any separable state may be extended to a \textit{cq}-state by purification of the 
%states in $R$ to some larger Hilbert space $RR'$.  The monotonicity of the conditional spectral entropy $\overline{S}(A|RR') 
%\leq \overline{S}(A|R)$ then gives the required bound.
\end{proof}

\section{The regularized entanglement of formation}
\label{asymp}

The application of the main result to the case of multiple copies of a single bipartite state provides a new proof of the equivalence \cite{hayden00} between the regularized entanglement of formation $E^{\infty}_F(\rho_{AB})$ (\ref{eof_reg}), of a bipartite state $\rho_{AB}$, and its entanglement cost 
$E_C(\rho_{AB})$.

First note that as the entanglement of formation is a bounded non-increasing function of $n$ we have $\inf_n \frac{1}{n}E_F(\rho_{AB}^{\otimes n}) = \lim_{n\rightarrow \infty} \frac{1}{n}E_F(\rho_{AB}^{\otimes n}) = E^{\infty}_F(\rho_{AB})$.
Consider a sequence $\hrho_{AB}=\{\rho^n_{AB}\}_{n=1}^\infty$ of a bipartite
states. For any state $\rho^n_{AB}$ in the sequence, let $S(A_n|B_n)_{\rho^n_{AB}}$ denote
the conditional entropy:
$$S(A_n|B_n)_{\rho^n_{AB}} = S(\rho^n_{AB}) - S(\rho^n_{B}).$$
From results in \cite{hayashi03} it can be shown that the conditional entropy rate of the sequence is bounded above by the sup-conditional spectral entropy 
rate:
\be
E^{\infty}_F(\rho_{AB})=\limsup_{n\rightarrow \infty} \frac{1}{n}S(A_n|B_n) \leq \overline{S}(A|B)
\label{up}
\ee
Thus, for any sequence of \textit{cq}-states 
$\hvarrho_{RAB}= \{\varrho_{RAB}^n\}_{n=1}^\infty$ on $RAB$, which reduce to product sequences $\hrho_{AB} = \{ \varrho^{\otimes n} \}_{n=1}^{\infty}$ on $AB$, we have from (\ref{eof}) and (\ref{up}) 
\begin{equation}
\inf_n \frac{1}{n}E_F(\varrho^{\otimes n}) \leq \liminf_{n\rightarrow \infty} \frac{1}{n}S(A_n|R_n)_{\varrho_{RA}^n}\leq \overline{S}(A|R), \nonumber
\end{equation}
where $\overline{S}(A|R)$ denotes the sup-conditional spectral entropy 
rate defined in (\ref{supar}). 

For the reverse inequality we simply construct states of block size $m$ on $RAB$ such that $\omega^{mn} = ( \sum_i p^{(m)}_i |i_R\rangle \langle i_R| \otimes |\phi^m_i\rangle \langle \phi_i^m|_{AB})^{\otimes \lfloor n/m \rfloor} \otimes \sigma_{RAB}$, where $\sigma$ is an asymptotically irrelevant buffer state whenever $m$ does not divide $n$.  Using the chain rule \cite{bowen06a} $\overline{S}(A|R) \leq \overline{S}(RA)-\underline{S}(R)$, the definitions of 
$\overline{S}(RA)$ and $\underline{S}(R)$, and (\ref{stein}), we obtain 
$$\overline{S}(A|R) \leq \frac{1}{m}\big( 
S(\omega_{RA}^{mm}) - S(\omega_{R}^{mm})\big) = \frac{1}{m}\sum_i p^{(m)}_i S(\omega^{i,A}_m),$$ for $\omega_{A,i}^m = \mathrm{Tr}_B |\phi_i^m\rangle 
\langle \phi_i^m|_{AB}$.  Taking the infimum over both $m$ and decompositions then implies
\begin{equation}
E_C(\rho_{AB}) = E^{\infty}_F(\rho_{AB})
\end{equation}
for product sequences, and hence the regularized entanglement of formation for a bipartite state is equal to its entanglement cost.

\section*{Acknowledgments}
This work is part of the QIP-IRC supported by the European Community's Seventh Framework Programme
(FP7/2007-2013) under grant agreement number 213681.
%\end{acknowledgments}

% Create the reference section using BibTeX:
%\bibliography{references.bib}
%

\end{document}